\documentclass[fleqn]{llncs}
\usepackage{version}
\usepackage{dlds}

\title{Data Linkage Dynamics with Shedding%
       \thanks{This research has been partly carried out in the
               framework of the Jacquard-project Symbiosis, which is
               funded by the Netherlands Organisation for Scientific
               Research (NWO).}}
\author{J.A. Bergstra \and C.A. Middelburg}
\institute{Informatics Institute, University of Amsterdam \\
           Science Park~107, 1098~XG Amsterdam, the Netherlands \\
           \email{J.A.Bergstra@uva.nl,C.A.Middelburg@uva.nl}}

\begin{document}

\maketitle

\begin{abstract}
We study shedding in the setting of data linkage dynamics, a simple
model of computation that bears on the use of dynamic data structures
in programming.
Shedding is complementary to garbage collection.
With shedding, each time a link to a data object is updated by a
program, it is determined whether or not the link will possibly be used
once again by the program, and if not the link is automatically removed.
Thus, everything is made garbage as soon as it can be viewed as garbage.
By that, the effectiveness of garbage collection becomes maximal.
\begin{keywords}
data linkage dynamics, shedding, forecasting service.
\end{keywords}
\begin{classcode}
D.3.3, D.4.2, F.1.1, F.3.3.
\end{classcode}
\end{abstract}

\section{Introduction}
\label{sect-intro}

This paper is a sequel to~\cite{BM08d}.
In that paper, we presented an algebra, called data linkage algebra, of
which the elements are intended for modelling the states of computations
in which dynamic data structures are involved.
We also presented a simple model of computation, called data linkage
dynamics, in which states of computations are modelled as elements of
data linkage algebra and state changes take place by means of certain
actions.
Data linkage dynamics includes the following features to reclaim
garbage: full garbage collection, restricted garbage collection (as if
reference counts are used), safe disposal of potential garbage, and
unsafe disposal of potential garbage.

In the current paper, we add shedding to the features of data linkage
dynamics.
This feature is complementary to the garbage collection features of data
linkage dynamics.
Roughly speaking, shedding works as follows: each time a link to a data
object is updated by a program, it is determined whether or not the link
will possibly be used once again by the program, and if not the link is
automatically removed.
In this way, everything is made garbage as soon as it can be taken for
garbage.
The point of shedding is that by this, the effectiveness of garbage
collection becomes maximal.

In the sixties of the previous century, when the first list-processing
languages came up, three basic garbage collection techniques have been
proposed: reference counting (see e.g.~\cite{GHG60a,Col60a}), marking
(see e.g.~\cite{McC60a,SW67a}), and copying (see
e.g.~\cite{Min63a,FY69a}).
The garbage collection techniques that have been proposed in the
seventies and eighties of the previous century are mainly incremental
and parallel variants of the three basic techniques (see
e.g.~\cite{Bak78a,LH83a} and~\cite{Ste75a,KS77a,DLMSS78a},
respectively), which are intended to avoid substantial interruption due
to garbage collection, and conservative and tag-free variants of the
three basic techniques (see e.g.~\cite{BW88a} and~\cite{App89,Gol91a},
respectively), which are intended to perform garbage collection more
efficient.
All the garbage collection techniques proposed in those times collect
only data objects that are no longer reachable by a series of links.
In the next paragraph, we will use the term ``standard garbage
collection techniques'' to refer to this group of garbage collection
techniques.

Owing to the growing use of dynamic data structures in programming, the
effectiveness of garbage collection techniques becomes increasingly more
important since the nineties of the previous century.
It has been confirmed by recent empirical studies that standard garbage
collection techniques actually leave a lot of garbage uncollected (see
e.g.~\cite{SKS00a,SKS01a,HDH02a}).
For the greater part, recently proposed garbage collection techniques
that are intended to be more effective than standard garbage collection
techniques turn out to make use of approximations of shedding.
The approximations are obtained by means of information about future
uses of links coming from static program analysis.
The information is either directly provided to an adapted standard
garbage collector (see e.g.~\cite{ADM98a}) or used to transform the
program in question such that data objects become unreachable as soon as
some safety property holds according to the information (see
e.g.~\cite{SYKS05a,KSK07a}).
In the latter case, the safety property used differs from one proposal
to another, can in all cases be improved by taking into account that the
number of data objects that can exist at the same time is bounded, and
is in all cases at best weakly justified by a precise semantics of the
programming language supposed to be used.

Our study of shedding arises from the work on ``nullifying dead links''
presented in~\cite{KSK07a}.
That work concerns the removal of links that will not possibly be used
once again by means of static program analysis and program
transformation.
In our study of shedding, different from the study in~\cite{KSK07a},
the semantic effects of the fact that the number of data objects that
can exist at the same time is always bounded are taken into account.

The view is taken that the behaviours exhibited by programs on execution
are threads as considered in basic thread algebra.%
\footnote
{In~\cite{BL02a}, basic thread algebra is introduced under the name
 basic polarized process algebra.
 Prompted by the development of thread algebra~\cite{BM04c}, which is a
 design on top of it, basic polarized process algebra has been renamed
 to basic thread algebra.
}
A thread proceeds by performing actions in a sequential fashion.
A thread may perform an action for the purpose of interacting with a
service that takes the action as a command to be processed.
The processing of the action results in a state change and a reply.
In the setting of basic thread algebra, the use mechanism has been
introduced to allow for this kind of interaction.
The state changes and replies that result from performing the actions
of data linkage dynamics can be achieved by means of a service.

In~\cite{BM08d}, it was explained how basic thread algebra can be
combined with data linkage dynamics by means of the use mechanism in
such a way that the whole can be used for studying issues concerning the
use of dynamic data structures in programming.
For a clear apprehension of data linkage dynamics as presented
in that paper, such a combination is not needed.
This is different for shedding, because it cannot be explained without
reference to program behaviours.
In the current paper, we adapt the data linkage dynamics services
involved in the combination described in~\cite{BM08d} to explain
shedding.
For the adapted data linkage dynamics services, shedding happens to be a
matter close to reflection on themselves.
Moreover, the adapted data linkage dynamics services are services of
which the state changes and replies may depend on how the thread that
performs the actions being processed will proceed.
That is why we also introduce a generalization of the use mechanism to
such forecasting services.

This paper is organized as follows.
First, we review data linkage algebra, data linkage dynamics and basic
thread algebra
(Sections~\ref{sect-DLA}, \ref{sect-DLD}, and~\ref{sect-BTA}).
Next, we present the use mechanism for forecasting services and explain
how basic thread algebra can be combined with data linkage dynamics by
means of that use mechanism
(Sections~\ref{sect-TSI} and~\ref{sect-comb-TA-DLD}).
After that, we introduce the shedding feature and adapt the data linkage
dynamics services involved in the combination described before such that
they support shedding
(Sections~\ref{sect-shedding}, \ref{sect-DLDSm}, and~\ref{sect-DLDSs}).
Then, we illustrate shedding by means of some examples
(Section~\ref{sect-examples}).
Finally, we make some concluding remarks (Section~\ref{sect-concl}).

\section{Data Linkage Algebra}
\label{sect-DLA}

In this section, we review the algebraic theory \DLA\ (Data Linkage
Algebra).
The elements of the initial algebra of \DLA\ can serve for the states of
computations in which dynamic data structures are involved.

In \DLA, it is assumed that a fixed but arbitrary finite set $\Spot$ of
\emph{spots}, a fixed but arbitrary finite set $\Field$ of
\emph{fields}, a fixed but arbitrary finite set $\AtObj$ of \emph{atomic
objects}, and a fixed but arbitrary finite set $\Value$ of \emph{values}
have been given.

\DLA\ has one sort: the sort $\DaLi$ of \emph{data linkages}.
To build terms of sort $\DaLi$, \BTA\ has the following constants and
operators:
\begin{itemize}
\item
for each $s \in \Spot$ and $a \in \AtObj$,
the \emph{spot link} constant $\const{\slink{s}{a}}{\DaLi}$;
\item
for each $a \in \AtObj$ and $f \in \Field$,
the \emph{partial field link} constant $\const{\pflink{a}{f}}{\DaLi}$;
\item
for each $a,b \in \AtObj$ and $f \in \Field$,
the \emph{field link} constant $\const{\flink{a}{f}{b}}{\DaLi}$;
\item
for each $a \in \AtObj$ and $n \in \Value$,
the \emph{value association} constant $\const{\valass{a}{n}}{\DaLi}$;
\item
the \emph{empty data linkage} constant $\const{\emptydl}{\DaLi}$;
\item
the binary \emph{data linkage combination} operator
$\funct{\dlcom}{\DaLi \x \DaLi}{\DaLi}$;
\item
the binary \emph{data linkage overriding combination} operator
$\funct{\dlori}{\DaLi \x \DaLi}{\DaLi}$.
\end{itemize}
Terms of sort $\DaLi$ are built as usual.
Throughout the paper, we assume that there are infinitely many variables
of sort $\DaLi$, including $X$, $Y$, $Z$.
We use infix notation for data linkage combination and data linkage
overriding combination.

Let $L$ and $L'$ be closed \DLA\ terms.
Then the constants and operators of \DLA\ can be explained as follows:
\begin{itemize}
\item
$\slink{s}{a}$ is the atomic data linkage that consists of a link via
spot $s$ to atomic object $a$;
\item
$\pflink{a}{f}$ is the atomic data linkage that consists of a partial
link from atomic object $a$ via field $f$;
\item
$\flink{a}{f}{b}$ is the atomic data linkage that consists of a link
from atomic object~$a$ via field $f$ to atomic object $b$;
\item
$\valass{a}{n}$ is the atomic data linkage that consists of an
association of the value~$n$ with atomic object $a$;
\item
$\emptydl$ is the data linkage that does not contain any atomic data
linkage;
\item
$L \dlcom L'$ is the union of the data linkages $L$ and $L'$;
\item
$L \dlori L'$ differs from $L \dlcom L'$ as follows:
\begin{itemize}
\item
if $L$ contains spot links via spot $s$ and $L'$ contains spot links via
spot $s$, then the former links are overridden by the latter ones;
\item
if $L$ contains partial field links and/or field links from atomic
object $a$ via field $f$ and $L'$ contains partial field links and/or
field links from atomic object $a$ via field $f$, then the former
partial field links and/or field links are overridden by the latter
ones;
\item
if $L$ contains value associations with atomic object $a$ and $L'$
contains value associations with atomic object $a$, then the former
value associations are overridden by the latter ones.
\end{itemize}
\end{itemize}

The axioms of \DLA\ are given in Table~\ref{axioms-DLA}.%
\begin{table}[!t]
\caption{Axioms of \DLA}
\label{axioms-DLA}
\begin{eqntbl}
\begin{seqncol}
X \dlcom Y = Y \dlcom X \\
X \dlcom (Y  \dlcom Z) = (X \dlcom Y)  \dlcom Z \\
X \dlcom X = X \\
X \dlcom \emptydl = X
\eqnsep
\emptydl \dlori X = X \\
X \dlori \emptydl = X \\
X \dlori (Y \dlcom Z) = (X \dlori Y) \dlcom (X \dlori Z) \\
(X \dlcom \slinkp{s}{a}) \dlori \slinkp{s}{b} =
X \dlori \slinkp{s}{b} \\
(X \dlcom \pflinkp{a}{f}) \dlori \pflinkp{a}{f} =
X \dlori \pflinkp{a}{f} \\
(X \dlcom \flinkp{a}{f}{b}) \dlori \pflinkp{a}{f} =
X \dlori \pflinkp{a}{f} \\
(X \dlcom \pflinkp{a}{f}) \dlori \flinkp{a}{f}{b} =
X \dlori \flinkp{a}{f}{b} \\
(X \dlcom \flinkp{a}{f}{b}) \dlori \flinkp{a}{f}{c} =
X \dlori \flinkp{a}{f}{c} \\
(X \dlcom \valass{a}{n}) \dlori \valass{a}{m} =
X \dlori \valass{a}{m} \\
(X \dlcom \slinkp{s}{a}) \dlori \slinkp{t}{b} =
(X \dlori \slinkp{t}{b}) \dlcom \slinkp{s}{a}
 & \mif s \neq t \\
(X \dlcom \pflinkp{a}{f}) \dlori \slinkp{s}{b} =
(X \dlori \slinkp{s}{b}) \dlcom \pflinkp{a}{f} \\
(X \dlcom \flinkp{a}{f}{b}) \dlori \slinkp{s}{c} =
(X \dlori \slinkp{s}{c}) \dlcom \flinkp{a}{f}{b} \\
(X \dlcom \valass{a}{n}) \dlori \slinkp{s}{b} =
(X \dlori \slinkp{s}{b}) \dlcom \valass{a}{n} \\
(X \dlcom \slinkp{s}{a}) \dlori \pflinkp{b}{f} =
(X \dlori \pflinkp{b}{f}) \dlcom \slinkp{s}{a} \\
(X \dlcom \pflinkp{a}{f}) \dlori \pflinkp{b}{g} =
(X \dlori \pflinkp{b}{g}) \dlcom \pflinkp{a}{f}
 & \mif a \neq b \lor f \neq g \\
(X \dlcom \flinkp{a}{f}{b}) \dlori \pflinkp{c}{g} =
(X \dlori \pflinkp{c}{g}) \dlcom \flinkp{a}{f}{b}
 & \mif a \neq c \lor f \neq g \\
(X \dlcom \valass{a}{n}) \dlori \pflinkp{b}{f} =
(X \dlori \pflinkp{b}{f}) \dlcom \valass{a}{n} \\
(X \dlcom \slinkp{s}{a}) \dlori \flinkp{b}{f}{c} =
(X \dlori \flinkp{b}{f}{c}) \dlcom \slinkp{s}{a} \\
(X \dlcom \pflinkp{a}{f}) \dlori \flinkp{b}{g}{c} =
(X \dlori \flinkp{b}{g}{c}) \dlcom \pflinkp{a}{f}
 & \mif a \neq b \lor f \neq g \\
(X \dlcom \flinkp{a}{f}{b}) \dlori \flinkp{c}{g}{d} =
(X \dlori \flinkp{c}{g}{d}) \dlcom \flinkp{a}{f}{b}
 & \mif a \neq c \lor f \neq g \\
(X \dlcom \valass{a}{n}) \dlori \flinkp{b}{f}{c} =
(X \dlori \flinkp{b}{f}{c}) \dlcom \valass{a}{n} \\
(X \dlcom \slinkp{s}{a}) \dlori \valass{b}{n} =
(X \dlori \valass{b}{n}) \dlcom \slinkp{s}{a} \\
(X \dlcom \pflinkp{a}{f}) \dlori \valass{b}{n} =
(X \dlori \valass{b}{n}) \dlcom \pflinkp{a}{f} \\
(X \dlcom \flinkp{a}{f}{b}) \dlori \valass{c}{n} =
(X \dlori \valass{c}{n}) \dlcom \flinkp{a}{f}{b} \\
(X \dlcom \valass{a}{n}) \dlori \valass{b}{m} =
(X \dlori \valass{b}{m}) \dlcom \valass{a}{n}
 & \mif a \neq b
\end{seqncol}
\end{eqntbl}
\end{table}
In this table, $s$ and $t$ stand for arbitrary spots from $\Spot$,
$f$ and $g$ stand for arbitrary fields from $\Field$,
$a$, $b$, $c$ and $d$ stand for arbitrary atomic objects from $\AtObj$,
and $n$ and $m$ stand for arbitrary values from $\Value$.

The set $\cB$ of \emph{basic terms} over \DLA\ is inductively defined by
the following rules:
\begin{itemize}
\item
$\emptydl \in \cB$;
\item
if $s \in \Spot$ and $a \in \AtObj$, then $\slink{s}{a} \in \cB$;
\item
if $a \in \AtObj$ and $f \in \Field$, then $\pflink{a}{f} \in \cB$;
\item
if $a,b \in \AtObj$ and $f \in \Field$, then $\flink{a}{f}{b} \in \cB$;
\item
if $a \in \AtObj$ and $n \in \Value$, then $\valass{a}{n} \in \cB$;
\item
if $L_1,L_2 \in \cB$, then $L_1 \dlcom L_2 \in \cB$.
\end{itemize}
\begin{theorem}
\label{theorem-elimination}
For all closed \DLA\ terms $L$, there exists a basic term $L' \in \cB$
such that $L = L'$ is derivable from the axioms of \DLA.
\end{theorem}
\begin{proof}
See Theorem~1 in~\cite{BM08d}.
\end{proof}

We are only interested in the initial model of \DLA.
We write $\DL$ for the set of all elements of the initial model of
$\DLA$.
$\DL$ consists of the equivalence classes of basic terms over \DLA\ with
respect to the equivalence induced by the axioms of \DLA.
In other words, modulo equivalence, $\cB$ is $\DL$.
Henceforth, we will identify basic terms over \DLA\ and their
equivalence classes.

\section{Data Linkage Dynamics}
\label{sect-DLD}

\DLD\ (Data Linkage Dynamics) is a simple model of computation that
bears on the use of dynamic data structures in programming.
It comprises states, basic actions, and the state changes and
replies that result from performing the basic actions.
The states of \DLD\ are data linkages.
In this section, we give an informal explanation of the basic actions of
\DLD\ to structure data dynamically.
The basic actions of \DLD\ to deal with values found in dynamically
structured data, as well as some actions related to reclaiming garbage,
are not explained.
For a comprehensive presentation of \DLD, the reader is referred
to~\cite{BM08d}.

Like in \DLA, it is assumed that a fixed but arbitrary finite set
$\Spot$ of spots, a fixed but arbitrary finite set $\Field$ of fields,
and a fixed but arbitrary finite set $\AtObj$ of atomic objects have
been given.
It is also assumed that a fixed but arbitrary \emph{choice}
function $\funct{\cf}{(\setof{(\AtObj)} \diff \emptyset)}{\AtObj}$ such
that, for all $A \in \setof{(\AtObj)} \diff \emptyset$, $\cf(A) \in A$
has been given.
The function $\cf$ is used whenever a fresh atomic object must be
obtained.

Below, we will informally explain the features of \DLD\ to structure data
dynamically.
When speaking informally about a state $L$ of \DLD, we say:
\begin{itemize}
\item
if there exists a unique atomic object $a$ for which $\slink{s}{a}$ is
contained in $L$, \emph{the content of spot $s$} instead of the unique
atomic object $a$ for which $\slink{s}{a}$ is contained in $L$;
\item
\emph{the fields of atomic object $a$} instead of the set of all fields
$f$ such that either $\pflink{a}{f}$ is contained in $L$ or there
exists an atomic object $b$ such that $\flink{a}{f}{b}$ is contained in
$L$;
\item
if there exists a unique atomic object $b$ for which $\flink{a}{f}{b}$
is contained in $L$, \emph{the content of field $f$ of atomic object
$a$} instead of the unique atomic object $b$ for which $\flink{a}{f}{b}$
is contained in $L$.
\end{itemize}
In the case where the uniqueness condition is met, the spot or field
concerned is called \emph{locally deterministic}.

\DLD\ has the following basic actions to structure data dynamically:
\begin{itemize}
\item
for each $s \in \Spot$,
a \emph{get fresh atomic object action} $\getatobj{s}$;
\item
for each $s,t \in \Spot$, a \emph{set spot action} $\setspot{s}{t}$;
\item
for each $s \in \Spot$, a \emph{clear spot action} $\clrspot{s}$;
\item
for each $s,t \in \Spot$,
an \emph{equality test action} $\equaltst{s}{t}$;
\item
for each $s \in \Spot$,
an \emph{undefinedness test action} $\undeftst{s}$;
\item
for each $s \in \Spot$ and $f \in \Field$,
a \emph{add field action} $\addfield{s}{f}$;
\item
for each $s \in \Spot$ and $f \in \Field$,
a \emph{remove field action} $\rmvfield{s}{f}$;
\item
for each $s \in \Spot$ and $f \in \Field$,
a \emph{has field action} $\hasfield{s}{f}$;
\item
for each $s,t \in \Spot$ and $f \in \Field$,
a \emph{set field action} $\setfield{s}{f}{t}$;
\item
for each $s \in \Spot$ and $f \in \Field$,
a \emph{clear field action} $\clrfield{s}{f}$;
\item
for each $s,t \in \Spot$ and $f \in \Field$,
a \emph{get field action} $\getfield{s}{t}{f}$.
\end{itemize}

If only locally deterministic spots and fields are involved, these
actions can be explained as follows:
\begin{itemize}
\item
$\getatobj{s}$:
if a fresh atomic object can be allocated, then the content of spot $s$
becomes that fresh atomic object and the reply is $\True$; otherwise,
nothing changes and the reply is $\False$;
\item
$\setspot{s}{t}$:
the content of spot $s$ becomes the same as the content of spot $t$
and the reply is $\True$;
\item
$\clrspot{s}$:
the content of spot $s$ becomes undefined and the reply is $\True$;
\item
$\equaltst{s}{t}$:
if the content of spot $s$ equals the content of spot $t$, then
nothing changes and the reply is $\True$; otherwise, nothing changes and
the reply is $\False$;
\item
$\undeftst{s}$:
if the content of spot $s$ is undefined, then nothing changes and the
reply is $\True$; otherwise, nothing changes and the reply is $\False$;
\item
$\addfield{s}{f}$:
if the content of spot $s$ is an atomic object and $f$ does not yet
belong to the fields of that atomic object, then $f$ is added (with
undefined content) to the fields of that atomic object and the reply is
$\True$; otherwise, nothing changes and the reply is $\False$;
\item
$\rmvfield{s}{f}$:
if the content of spot $s$ is an atomic object and $f$ belongs to the
fields of that atomic object, then $f$ is removed from the fields of
that atomic object and the reply is $\True$; otherwise, nothing changes
and the reply is $\False$;
\item
$\hasfield{s}{f}$:
if the content of spot $s$ is an atomic object and $f$ belongs to the
fields of that atomic object, then nothing changes and the reply is
$\True$; otherwise, nothing changes and the reply is $\False$;
\item
$\setfield{s}{f}{t}$:
if the content of spot $s$ is an atomic object and $f$ belongs to the
fields of that atomic object, then the content of that field becomes the
same as the content of spot $t$ and the reply is $\True$; otherwise,
nothing changes and the reply is $\False$;
\item
$\clrfield{s}{f}$:
if the content of spot $s$ is an atomic object and $f$ belongs to the
fields of that atomic object, then the content of that field becomes
undefined and the reply is $\True$; otherwise, nothing changes and the
reply is $\False$;
\item
$\getfield{s}{t}{f}$:
if the content of spot $t$ is an atomic object and $f$ belongs to the
fields of that atomic object, then the content of spot $s$ becomes the
same as the content of that field and the reply is $\True$; otherwise,
nothing changes and the reply is~$\False$.
\end{itemize}
In the explanation given above, wherever we say that the content of a
spot or field becomes the same as the content of another spot or field,
this is meant to imply that the former content becomes undefined if the
latter content is undefined.
If not only locally deterministic spots and fields are involved in
performing an action, there is no state change and the reply is
$\False$.

Atomic objects that are not reachable via spots and fields can be
reclaimed.
Reclamation of unreachable atomic objects is relevant because the set
$\AtObj$ of atomic objects is finite.
In~\cite{BM08d}, we introduce various ways to achieve reclamation of
unreachable atomic objects.
In this section, we mention only one of the reclamation-related actions:
the \emph{full garbage collection action} $\fgc$.
By performing this action, all unreachable atomic objects are reclaimed.
The reply that results from performing this action is always $\True$.

We write $\Act_\DLD$ for the set of all basic actions of \DLD.

In~\cite{BM08d}, we describe the state changes and replies that result
from performing the basic actions of \DLD\ by means of a term rewrite
system with rule priorities~\cite{BBKW89a}.
For that purpose, a unary \emph{effect} operator $\effop{\alpha}$ and a
unary \emph{yield} operator $\yldop{\alpha}$ are introduced for each
basic action $\alpha \in \Act_\DLD$.
The intuition is that these operators stand for operations that give,
for each state $L$, the state and reply, respectively, that result from
performing basic action $\alpha$ in state $L$.

\section{Basic Thread Algebra}
\label{sect-BTA}

In this section, we review the algebraic theory \BTA\ (Basic Thread
Algebra), a form of process algebra which is tailored to the description
and analysis of the behaviours of sequential programs under execution.
The behaviours concerned are called \emph{threads}.

In \BTA, it is assumed that a fixed but arbitrary finite set $\BAct$ of
\emph{basic actions}, with $\Tau \not\in \BAct$, has been given.
We write $\BActTau$ for $\BAct \union \set{\Tau}$.
The members of $\BActTau$ are referred to as \emph{actions}.

Threads proceed by performing actions in a sequential fashion.
Each basic action performed by a thread is taken as a command to be
processed by some service provided by the execution environment of the
thread.
The processing of a command may involve a change of state of the service
concerned.
At completion of the processing of the command, the service returns a
reply value $\True$ or $\False$ to the thread concerned.

\BTA\ has one sort: the sort $\Thr$ of \emph{threads}.
To build terms of sort $\Thr$, \BTA\ has the following constants and
operators:
\begin{itemize}
\item
the \emph{deadlock} constant $\const{\DeadEnd}{\Thr}$;
\item
the \emph{termination} constant $\const{\Stop}{\Thr}$;
\item
for each $\alpha \in \BActTau$,
the binary \emph{postconditional composition} operator
$\funct{\pcc{\ph}{\alpha}{\ph}}{\Thr \x \Thr}{\Thr}$.
\end{itemize}
Terms of sort $\Thr$ are built as usual (see e.g.~\cite{Wir90a,ST99a}).
Throughout the paper, we assume that there are infinitely many variables
of sort $\Thr$, including $x,y,z$.

We use infix notation for postconditional composition.
We introduce \emph{action prefixing} as an abbreviation:
$\alpha \bapf p$, where $p$ is a term of sort $\Thr$, abbreviates
$\pcc{p}{\alpha}{p}$.

Let $p$ and $q$ be closed terms of sort $\Thr$ and
$\alpha \in \BActTau$.
Then $\pcc{p}{\alpha}{q}$ will perform action $\alpha$, and after that
proceed as $p$ if the processing of $\alpha$ leads to the reply $\True$
(called a positive reply), and proceed as $q$ if the processing of
$\alpha$ leads to the reply $\False$ (called a negative reply).
The action $\Tau$ plays a special role.
It is a concrete internal action: performing $\Tau$ will never lead to a
state change and always lead to a positive reply, but notwithstanding
all that its presence matters.

\BTA\ has only one axiom.
This axiom is given in Table~\ref{axioms-BTA}.%
\begin{table}[!t]
\caption{Axiom of \BTA}
\label{axioms-BTA}
\begin{eqntbl}
\begin{axcol}
\pcc{x}{\Tau}{y} = \pcc{x}{\Tau}{x}                      & \axiom{T1}
\end{axcol}
\end{eqntbl}
\end{table}

Each closed \BTA\ term of sort $\Thr$ denotes a finite thread, i.e.\ a
thread of which the length of the sequences of actions that it can
perform is bounded.
Guarded recursive specifications give rise to infinite threads.

A \emph{guarded recursive specification} over \BTA\ is a set of
recursion equations $E = \set{X = p_X \where X \in V}$, where $V$ is a
set of variables of sort $\Thr$ and each $p_X$ is a term of the form
$\DeadEnd$, $\Stop$ or $\pcc{p}{\alpha}{q}$ with $p$ and $q$ \BTA\ terms
of sort $\Thr$ that contain only variables from $V$.
We write $\vars(E)$ for the set of all variables that occur on the
left-hand side of an equation in $E$.
We are only interested in models of \BTA\ in which guarded recursive
specifications have unique solutions, such as the projective limit model
of \BTA\ presented in~\cite{BB03a}.

We extend \BTA\ with guarded recursion by adding constants for solutions
of guarded recursive specifications and axioms concerning these
additional constants.
For each guarded recursive specification $E$ and each $X \in \vars(E)$,
we add a constant of sort $\Thr$ standing for the unique solution of $E$
for $X$ to the constants of \BTA.
The constant standing for the unique solution of $E$ for $X$ is denoted
by $\rec{X}{E}$.
Moreover, we add the axioms for guarded recursion given in
Table~\ref{axioms-REC} to \BTA,%
\begin{table}[!t]
\caption{Axioms for guarded recursion}
\label{axioms-REC}
\begin{eqntbl}
\begin{saxcol}
\rec{X}{E} = \rec{t_X}{E} & \mif X \!=\! t_X \in E       & \axiom{RDP}
\\
E \limpl X = \rec{X}{E} & \mif X \in \vars(E)          & \axiom{RSP}
\end{saxcol}
\end{eqntbl}
\end{table}
where we write $\rec{t_X}{E}$ for $t_X$ with, for all $Y \in \vars(E)$,
all occurrences of $Y$ in $t_X$ replaced by $\rec{Y}{E}$.%
\footnote
{Throughout the paper, we use the symbol $\limpl$ for implication.}
In this table, $X$, $t_X$ and $E$ stand for an arbitrary variable of
sort $\Thr$, an arbitrary \BTA\ term of sort $\Thr$ and an arbitrary
guarded recursive specification over \BTA, respectively.
Side conditions are added to restrict the variables, terms and guarded
recursive specifications for which $X$, $t_X$ and $E$ stand.

Henceforth, we write \BTA+\REC\ for \BTA\ extended with the constants
for solutions of guarded recursive specifications and axioms RDP and
RSP.
Moreover, we write $\TThr$ for the set of all closed terms of \BTA+\REC.

In the following definition, the interpretation of a postconditional
composition operator in a model of \BTA+\REC\ is denoted by the operator
itself.
Let $\fM$ be some model of \BTA+\REC, and let $p$ be an element from the
domain of $\fM$.
Then the set of \emph{residual threads} of $p$, written $\Res(p)$, is
inductively defined as follows:
\begin{itemize}
\item
$p \in \Res(p)$;
\item
if $\pcc{q}{a}{r} \in \Res(p)$, then $q \in \Res(p)$ and
$r \in \Res(p)$.
\end{itemize}
We say that $p$ is \emph{regular} if $\Res(p)$ is finite.

We are only interested in models of \BTA+\REC\ in which the solution of
a guarded recursive specification $E$ over \BTA\ is regular if and only
if $E$ is finite, such as the projective limit model presented
in~\cite{BB03a}.
Par abus de langage, a closed term of \BTA+\REC\ without occurrences of
constants $\rec{X}{E}$ for infinite $E$ will henceforth be called a
\emph{regular thread}.

\section{A Use Mechanism for Forecasting Services}
\label{sect-TSI}

A thread may perform an action for the purpose of interacting with a
service that takes the action as a command to be processed.
The processing of the action may involve a change of state of the
service and at completion of the processing of the action the service
returns a reply value to the thread.
In this section, we introduce a mechanism that is concerned with this
kind of interaction.
It is a generalization of the use mechanism introduced in~\cite{BM08d}
to forecasting services.
A forecasting service is a service of which the state changes and
replies may depend on how the thread that performs the actions being
processed will proceed.

It is assumed that a fixed but arbitrary finite set $\Foci$ of
\emph{foci} and a fixed but arbitrary finite set $\Meth$ of
\emph{methods} have been given.
Each focus plays the role of a name of some service provided by an
execution environment that can be requested to process a command.
Each method plays the role of a command proper.
For the set $\BAct$ of actions, we take the set
$\set{f.m \where f \in \Foci, m \in \Meth}$.
Performing an action $f.m$ is taken as making a request to the
service named $f$ to process command~$m$.

Recall that $\TThr$ stands for the set of all closed terms of \BTA+\REC.

A \emph{forecasting service} $H$ consists of
\begin{itemize}
\item
a set $S$ of \emph{states};
\item
an \emph{effect} function $\funct{\eff}{\Meth \x S \x \TThr}{S}$;
\item
a \emph{yield} function
$\funct{\yld}{\Meth \x S \x \TThr}{\set{\True,\False,\Blocked}}$;
\item
an \emph{initial state} $s_0 \in S$;
\end{itemize}
satisfying the following conditions:
\begin{ldispl}
\Exists{s \in S}
 {\Forall{m \in \Meth, p \in \TThr}{{}}}
\\ \quad
(\yld(m,s,p) = \Blocked \land
 \Forall{s' \in S}
  {(\yld(m,s',p) = \Blocked \limpl \eff(m,s',p) = s)})\;,
\eqnsep
\Forall{s\in S, m,m' \in \Meth, f \in \Foci, p,q \in \TThr}{{}}
\\ \quad
 (\yld(m,s,\Stop) = \Blocked \land
  \yld(m,s,\DeadEnd) = \Blocked \land {}
\\ \quad \phantom{(}
  \yld(m,s,\Tau \bapf p) = \Blocked \land
  (m \neq m' \limpl
   \yld(m,s,\pcc{p}{f.m'}{q}) = \Blocked))\;.
\end{ldispl}%
The set $S$ contains the states in which the services may be, and the
functions $\eff$ and $\yld$ give, for each method $m$, state $s$ and
thread $p$, the state and reply, respectively, that result from
processing $m$ in state $s$ if $p$ is the thread that makes the request
to process $m$.
In certain states, requests to process certain methods may be rejected.
$\Blocked$, which stands for blocked, is used to indicate this.

Given a forecasting service $H = \tup{S,\eff,\yld,s_0}$, a method
$m \in \Meth$ and a thread $p \in \TThr$:
\begin{itemize}
\item
the \emph{derived service} of $H$ after processing $m$ in the context
of $p$, written $\derive{m}H[p]$, is the forecasting service
$\tup{S,\eff,\yld,\eff(m,s_0,p)}$;
\item
the \emph{reply} of $H$ after processing $m$ in the context of $p$,
written $H[p](m)$, is $\yld(m,s_0,p)$.
\end{itemize}

A forecasting service $H = \tup{S,\eff,\yld,s_0}$ can be understood as
follows:
\begin{itemize}
\item
if thread $p$ makes a request to the service to process $m$ and
$H[p](m) \neq \Blocked$, then the request is accepted, the reply is
$H[p](m)$, and the service proceeds as $\derive{m}H[p]$;
\item
if thread $p$ makes a request to the service to process $m$ and
$H[p](m) = \Blocked$, then the request is rejected.
\end{itemize}
By the first condition on forecasting services, after a request has been
rejected by the service, it gets into a state in which any request will
be rejected.
By the second condition on forecasting services, any request that does
not correspond to the action being performed by thread $p$ is rejected.

In the case of a forecasting service $H = \tup{S,\eff,\yld,s_0}$,
the derived service and reply that result from processing a method may
depend on how the thread that makes the request to process that method
will proceed.
Hence the name forecasting service.
Henceforth, we will omit the qualification forecasting if no confusion
can arise with other kinds of services.

We introduce yet another sort: the sort $\Serv$ of \emph{services}.
However, we will not introduce constants and operators to build terms of
this sort.
We demand that the interpretation of the sort $\Serv$ in a model is a
set $\FSs$ of forecasting services such that for all $H \in \FSs$,
$\derive{m}H[p] \in \FSs$ for each $m \in \Meth$ and $p \in \TThr$.

We introduce the following additional operators:
\begin{itemize}
\item
for each $f \in \Foci$, the binary \emph{use} operator
$\funct{\use{\ph}{f}{\ph}}{\Thr \x \Serv}{\Thr}$.
\end{itemize}
We use infix notation for the use operators.

Intuitively, $\use{p}{f}{H}$ is the thread that results from processing
all actions performed by thread $p$ that are of the form $f.m$ by
service $H$.
When an action of the form $f.m$ performed by thread $p$ is processed by
service $H$, that action is turned into the internal action $\Tau$ and
postconditional composition is removed in favour of action prefixing on
the basis of the reply value produced.
In previous work, we sometimes opted for the alternative to conceal the
processed actions completely.
However, we experienced repeatedly in cases where this alternative
appeared to be appropriate at first that it turned out to impede
progress later.

The axioms for the use operators are given in Table~\ref{axioms-use}.%
\begin{table}[!t]
\caption{Axioms for use operators}
\label{axioms-use}
\begin{eqntbl}
\begin{saxcol}
\use{\Stop}{f}{H} = \Stop                            & & \axiom{TSU1} \\
\use{\DeadEnd}{f}{H} = \DeadEnd                      & & \axiom{TSU2} \\
\use{(\Tau \bapf p)}{f}{H} =
                          \Tau \bapf (\use{p}{f}{H}) & & \axiom{TSU3} \\
\use{(\pcc{p}{g.m}{q})}{f}{H} =
\pcc{(\use{p}{f}{H})}{g.m}{(\use{q}{f}{H})}
 & \mif f \neq g                                       & \axiom{TSU4} \\
\use{(\pcc{p}{f.m}{q})}{f}{H} =
\Tau \bapf (\use{p}{f}{\derive{m}H[\pcc{p}{f.m}{q}]})
         & \mif H[\pcc{p}{f.m}{q}](m) = \True    & \axiom{TSU5} \\
\use{(\pcc{p}{f.m}{q})}{f}{H} =
\Tau \bapf (\use{q}{f}{\derive{m}H[\pcc{p}{f.m}{q}]})
         & \mif H[\pcc{p}{f.m}{q}](m) = \False   & \axiom{TSU6} \\
\use{(\pcc{p}{f.m}{q})}{f}{H} = \DeadEnd
         & \mif H[\pcc{p}{f.m}{q}](m) = \Blocked & \axiom{TSU7}
\end{saxcol}
\end{eqntbl}
\end{table}
In this table, $f$ and $g$ stand for arbitrary foci from $\Foci$, $m$
stands for an arbitrary method from $\Meth$, and $p$ and $q$ stand for
arbitrary closed terms of sort $\Thr$.
$H$ ranges over the interpretation of sort $\Serv$.
Axioms TSU3 and TSU4 express that the action $\Tau$ and actions of
the form $g.m$, where $f \neq g$, are not processed.
Axioms TSU5 and TSU6 express that a thread is affected by a service
as described above when an action of the form $f.m$ is processed by the
service.
Axiom TSU7 expresses that deadlock takes place when an action to be
processed is not accepted.

Henceforth, we write \BTAuse\ for \BTA, taking the set
$\set{f.m \where f \in \Foci, m \in \Meth}$ for $\BAct$, extended with
the use operators and the axioms from Table~\ref{axioms-use}.

The use mechanism introduced in~\cite{BM04c} deals in essence with
forecasting services of which:
\begin{itemize}
\item
the set of states is the set of all sequences with elements from
$\Meth$;
\item
the derived service and reply that result from processing a method do
not depend on how the thread that makes the request to process that
method will proceed.
\end{itemize}
For these services, the use mechanism introduced in this section
coincides with the use mechanism introduced in~\cite{BM04c}.
The architecture-dependent services considered in~\cite{BBP07a}
can be looked upon as simple forecasting services.

\section{Thread Algebra and Data Linkage Dynamics Combined}
\label{sect-comb-TA-DLD}

The state changes and replies that result from performing the actions
of data linkage dynamics can be achieved by means of services.
In this short section, we explain how basic thread algebra can be
combined with data linkage dynamics by means of the use mechanism
introduced in Section~\ref{sect-TSI} such that the whole can be used for
studying issues concerning the use of dynamic data structures in
programming.
The services involved do not have a forecasting nature.
The adapted services needed to deal with shedding, which are described
in Section~\ref{sect-DLDSs}, have a forecasting nature.

Recall that $\DL$ stands for the set of all elements of the initial
model of $\DLA$, and recall that, for each $\alpha \in \Act_\DLD$,
$\effop{\alpha}$ and $\yldop{\alpha}$ stand for unary operations on
$\DL$ that give, for $L \in \DL$, the state and reply, respectively,
that result from performing basic action $\alpha$ in state $L$.
It is assumed that a blocking state $\undef \not\in \DL$ has been given.

Take $\Meth$ such that $\Act_\DLD \subseteq \Meth$.
Moreover, let $L \in \DL \union \set{\undef}$.
Then the \emph{data linkage dynamics service} with initial state $L$,
written $\DLDS(L)$, is the service
$\tup{\DL \union \set{\undef},\eff,\yld,L}$, where the functions $\eff$
and $\yld$ are the effect and yield functions satisfying the
(unconditional and conditional) equations in Table~\ref{eqns-DLDS}.%
\begin{table}[!t]
\caption{Definition of effect and yield functions for \DLD}
\label{eqns-DLDS}
\begin{eqntbl}
\begin{seqncol}
\eff(m,L,\pcc{p}{f.m}{q}) = \effop{m}(L) & \mif m \in \Act_\DLD
\\
\eff(m,L,\pcc{p}{f.m}{q}) = \undef       & \mif m \not\in \Act_\DLD
\eqnsep
\yld(m,L,\pcc{p}{f.m}{q}) = \yldop{m}(L) & \mif m \in \Act_\DLD
\\
\yld(m,L,\pcc{p}{f.m}{q}) = \Blocked     & \mif m \not\in \Act_\DLD
\eqnsep
\yld(m,L,p) = \Blocked \limpl \eff(m,L,p) = \undef
\end{seqncol}
\end{eqntbl}
\end{table}
Notice that, because of the conditions imposed on forecasting services
in Section~\ref{sect-TSI}, these equations characterize the effect and
yield functions uniquely.

By means of threads and the data linkage dynamics services introduced
above, we can give a precise picture of computations in which dynamic
data structures are involved.
Examples of such computations can be found in~\cite{BM08d}.

The combination of basic thread algebra and data linkage dynamics by
means of the use mechanism can be used for studying issues concerning
the use of dynamic data structures in programming at the level of
program behaviours.
A hierarchy of simple program notations rooted in \PGA\ is presented
in~\cite{BL02a}.
Included are program notations which are close to existing assembly
languages up to and including program notations that support structured
programming by offering a rendering of conditional and loop constructs.
Regular threads are taken as the behaviours of programs in those program
notations.
Together with one of the program notations, the combination of basic
thread algebra and data linkage dynamics can be used for studying issues
concerning the use of dynamic data structures in programming at the
level of programs.
We mention one such issue.
In general terms, the issue is whether we can do without garbage
collection by program transformation at the price of a linear increase
of the number of available atomic objects.
In~\cite{BM08d}, we phrase this issue precisely for one of the program
notation rooted in \PGA.

The notation for the basic actions of \DLD, makes the focus-method
notation $f.m$ less suitable in the case where $m$ is a basic action of
\DLD.
Therefore, we will henceforth mostly write $f(m)$ instead of $f.m$ if
$m \in \Act_\DLD$.

\section{The Shedding Feature}
\label{sect-shedding}

In this section, we introduce the shedding feature in the setting of
data linkage dynamics in an informal way.
In Section~\ref{sect-DLDSs}, we will adapt the data linkage dynamics
services introduced in Section~\ref{sect-comb-TA-DLD} to explain
shedding in a more precise way.

Roughly speaking, shedding works as follows: each time the content of a
spot or field is changed, it is determined whether or not the spot or
field will possibly be used once again, and if not its content is made
undefined.
If a spot or field is made undefined in this way, we say that it is
shed.
The use of a previously shed spot or field is called a shedding error.

The shedding feature is rather non-obvious.
Consider the thread
\begin{ldispl}
\dld(\getatobj{s}) \bapf
(\pcc{(\dld(\setspot{u}{s}) \bapf \Stop)}{\dld(\getatobj{t})}{\Stop})
\end{ldispl}%
and assume that the cardinality of $\AtObj$ is $1$.
If $s$ is not shed on performing $\getatobj{s}$, then a negative reply
is produced on performing $\getatobj{t}$ and the thread terminates
without having made use of $s$.
However, from this it cannot be concluded that $s$ could be shed on
performing $\getatobj{s}$ after all.
If $s$ would be shed on performing $\getatobj{s}$, a positive reply
would be produced on performing $\getatobj{t}$ and after that a shedding
error would occur.
This shows that shedding becomes paradoxical if we do not deal properly
with the fact that shedding of a spot or field influences whether or not
that spot or field will possibly be used once again.

In the light of this, it is of the utmost importance to have the right
criterion for shedding in mind:
\begin{quote}
a spot or field can safely be shed if it is not possible for the program
behaviour under consideration to evolve in the case where that spot or
field is shed, irrespective as to whether other spots and fields are
subsequently shed, in such a way that the first shedding error concerns
that spot or field.
\end{quote}
When speaking about applications of this criterion, shedding errors that
concern the spot or field to which the criterion is applied are called
\emph{primary} shedding errors and other shedding errors are called
\emph{secondary} shedding errors.

In Section~\ref{sect-comb-TA-DLD}, it was explained how basic thread
algebra can be combined with data linkage dynamics by means of the use
mechanism from Section~\ref{sect-TSI} in such a way that the whole can
be used for studying issues concerning the use of dynamic data
structures in programming.
For a clear apprehension of data linkage dynamics as presented in
Section~\ref{sect-DLD}, such a combination is not needed.
This is different for shedding: it cannot be explained without reference
to program behaviours.
In Section~\ref{sect-DLDSs}, we adapt the data linkage dynamics services
involved in the combination described in Section~\ref{sect-comb-TA-DLD}
to explain shedding.

For the adapted data linkage dynamics services, shedding happens to be a
matter close to reflection on itself.
Material to the adaptation is the above-mentioned criterion for shedding
a spot or field.
Instrumental in checking this criterion are the data linkage dynamics
services for a minor variation of \DLD.
It concerns services which support the mimicking of shedding.

\section{Mimicking of Shedding}
\label{sect-DLDSm}

The shedding supporting data linkage dynamics services, which will be
introduced in Section~\ref{sect-DLDSs}, check the criterion for shedding
adopted in Section~\ref{sect-shedding} by determining what would happen
if mimicking of shedding supporting data linkage dynamics services were
used.
In this section, we describe the mimicking of shedding supporting data
linkage dynamics services in question.

These services are data linkage dynamics services for a variation of
\DLD.
The variation concerned, referred to as \DLDm, differs from \DLD\ as
follows:
\begin{itemize}
\item
it has two additional atomic objects $\pso$ and $\sso$;
\item
for each $s \in \Spot$, it has two additional basic actions
$\setspot{s}{\pso}$ and $\setspot{s}{\sso}$;
\item
for each $s \in \Spot$ and $f \in \Field$, it has two additional basic
actions $\setfield{s}{f}{\pso}$ and $\setfield{s}{f}{\sso}$;
\item
on performing $\getatobj{s}$, the contents of spot $s$ never becomes
$\pso$ or $\sso$;
\item
on performing $\fgc$, $\pso$ and $\sso$ are never reclaimed.
\end{itemize}

If only locally deterministic spots and fields are involved, the
additional basic actions can be explained as follows:
\begin{itemize}
\item
$\setspot{s}{\pso}$:
the content of spot $s$ becomes $\pso$ and the reply is $\True$;
\item
$\setspot{s}{\sso}$:
the content of spot $s$ becomes $\sso$ and the reply is $\True$;
\item
$\setfield{s}{f}{\pso}$:
if the content of spot $s$ is an atomic object and $f$ belongs to the
fields of that atomic object, then the content of that field becomes
$\pso$ and the reply is $\True$; otherwise, nothing changes and the
reply is $\False$;
\item
$\setfield{s}{f}{\sso}$:
if the content of spot $s$ is an atomic object and $f$ belongs to the
fields of that atomic object, then the content of that field becomes
$\sso$ and the reply is $\True$; otherwise, nothing changes and the
reply is $\False$.
\end{itemize}
If not only locally deterministic spots and fields are involved in
performing an action, there is no state change and the reply is
$\False$.

The special atomic objects $\pso$ and $\sso$ are used as follows:
\begin{itemize}
\item
when checking of the criterion for some spot or field starts, the
shedding of that spot or field is mimicked by setting its content to
$\pso$;
\item
during checking, the shedding of another spot or field is mimicked by
setting its content to $\sso$.
\end{itemize}
If a spot or field is used whose content is $\pso$, a mimicked primary
shedding error is encountered and, if a spot or field is used whose
content is $\sso$, a mimicked secondary shedding error is encountered.

Different sets of spots, sets of fields or sets of atomic objects give
rise to different instances of \DLA.
The states of \DLD\ are the elements of the initial model of some
instance of \DLA.
Because of the two additional atomic objects, the states of \DLDm\ are
the elements of the initial model of another instance of \DLA.
Henceforth, we write $\DL$ for the set of all elements of the initial
model of former instance of \DLA\ and $\DLm$ for the set of all elements
of the initial model of latter instance of \DLA.
In~\cite{BM08d}, we describe the state changes and replies that result
from performing basic actions of \DLD\ by means of a term rewrite system
with rule priorities.
It is obvious how that term rewrite system must be adapted to obtain a
term rewrite system describing the state changes and replies that result
from performing basic actions of \DLDm.
For each basic action $\alpha$ of \DLDm, we write $\effopm{\alpha}$ and
$\yldopm{\alpha}$ for the effect and yield operators that go with
$\alpha$ in the latter term rewrite system.
Moreover, we write $A_\DLDm$ for the set of all basic actions of \DLDm.

Let $L \in \DLm \union \set{\undef}$.
Then the \emph{mimicking of shedding supporting data linkage dynamics
service} with initial state $L$, written $\DLDSm(L)$, is the service
$\tup{\DLm \union \set{\undef},\effm,\yldm,L}$, where the functions
$\effm$ and $\yldm$ are the effect and yield functions satisfying the
equations in Table~\ref{eqns-DLDSm}.%
\begin{table}[!t]
\caption{Definition of effect and yield functions for \DLD\ with
  mimicking of shedding}
\label{eqns-DLDSm}
\begin{eqntbl}
\begin{seqncol}
\effm(m,L,\pcc{p}{f.m}{q}) = \effopm{m}(L) & \mif m \in \Act_\DLDm
\\
\effm(m,L,\pcc{p}{f.m}{q}) = \undef        & \mif m \not\in \Act_\DLDm
\eqnsep
\yldm(m,L,\pcc{p}{f.m}{q}) = \yldopm{m}(L) & \mif m \in \Act_\DLDm
\\
\yldm(m,L,\pcc{p}{f.m}{q}) = \Blocked      & \mif m \not\in \Act_\DLDm
\eqnsep
\yldm(m,L,p) = \Blocked \limpl \effm(m,L,p) = \undef
\end{seqncol}
\end{eqntbl}
\end{table}

\section{Shedding Supporting Data Linkage Dynamics Services}
\label{sect-DLDSs}

In this section, we turn to the data linkage dynamic services that
support shedding themselves.

We assume that $\dld \in \Foci$.
It is supposed that requests to a shedding supporting data linkage
dynamics service to process basic actions of \DLD\ are always made using
the focus $\dld$.
We write $A^\sh_\DLD$ for the set of all basic actions of \DLD\ that are
of the form $\getatobj{s}$, $\setspot{s}{t}$, $\setfield{s}{f}{t}$ or
$\getfield{s}{t}{f}$.

In the definition of shedding supporting data linkage dynamics services
given below, we use an auxiliary function $\funct{\shv}{A_\DLD}{A_\DLD}$
and a set $\shok \subseteq \TThr \x \DL$.

The function $\shv$ gives, for each basic action of \DLD\ for changing
the content of a spot or field, the basic action of \DLD\ for making the
content of that spot or field undefined.
For each other basic action of \DLD, $\shv$ gives the basic action
itself.
The function $\shv$ is defined as follows:
\begin{ldispl}
\shv(\getatobj{s}) = (\clrspot{s})\;, \\
\shv(\setspot{s}{t}) = (\clrspot{s})\;, \\
\shv(\setfield{s}{f}{t}) = (\clrfield{s}{f})\;, \\
\shv(\getfield{s}{t}{f}) = (\clrspot{s})\;, \\
\shv(\alpha) = \alpha \hsp{.5} \mif \alpha \not\in A^\sh_\DLD\;.
\end{ldispl}%

In the definition of the set $\shok$, we use an auxiliary function
$\funct{\mshv}{\set{0,1,2} \x A_\DLD}{A_\DLDm}$ and, for each
$L \in \DLm$, sets $nosherr(L), \secsherr(L) \subseteq A_\DLD$.

The function $\mshv$ gives, for each natural number in the set
$\set{0,1,2}$ and each basic action of \DLD\ for changing the content of
a spot or field: the basic action itself if the number is $0$,
the basic action of \DLDm\ for making the content of that spot or field
$\pso$ if the number is $1$, and
the basic action of \DLDm\ for making the content of that spot or field
$\sso$ if the number is $2$.
For each other basic action of \DLD, $\mshv$ gives always the basic
action itself.
The function $\mshv$ is defined as follows:
\begin{ldispl}
\begin{geqns}
\mshv(0,\alpha) = \alpha\;, \\
\mshv(1,\getatobj{s}) = (\setspot{s}{\pso})\;, \\
\mshv(1,\setspot{s}{t}) = (\setspot{s}{\pso})\;, \\
\mshv(1,\setfield{s}{f}{t}) = (\setfield{s}{f}{\pso})\;, \\
\mshv(1,\getfield{s}{t}{f}) = (\setspot{s}{\pso})\;,
\end{geqns}
\qquad
\begin{geqns}
\mshv(i,\alpha) = \alpha \hfill \mif \alpha \not\in A^\sh_\DLD\;, \\
\mshv(2,\getatobj{s}) = (\setspot{s}{\sso})\;, \\
\mshv(2,\setspot{s}{t}) = (\setspot{s}{\sso})\;, \\
\mshv(2,\setfield{s}{f}{t}) = (\setfield{s}{f}{\sso})\;, \\
\mshv(2,\getfield{s}{t}{f}) = (\setspot{s}{\sso})\;.
\end{geqns}
\end{ldispl}%

For each $L \in \DLm$, the set $\nosherr(L)$ contains all basic actions
$\alpha \in A_\DLD$ whose use in state $L$ does not amount to a mimicked
shedding error and the set $\secsherr(L)$ contains all basic actions
$\alpha \in A_\DLD$ whose use in state $L$ amounts to a mimicked
secondary shedding error.
For each $L \in \DLm$, the set $\nosherr(L)$ is inductively defined as
follows:
\begin{itemize}
\item
$\getatobj{s},\, \clrspot{s} \in \nosherr(L)$;
\item
if $L \dlcom \slinkp{s}{a} = L$, $a \neq \pso$ and $a \neq \sso$, \\
then
$\setspot{t}{s},\, \undeftst{s},\,
 \addfield{s}{f},\, \rmvfield{s}{f},\, \hasfield{s}{f} \in \nosherr(L)$;
\item
if $L \dlcom \slinkp{s}{a} \dlcom \slinkp{t}{b} = L$, $a \neq \pso$,
$a \neq \sso$, $b \neq \pso$ and $b \neq \sso$, \\
then $\equaltst{s}{t},\, \setfield{s}{f}{t} \in \nosherr(L)$;
\item
if $L \dlcom \slinkp{s}{a} \dlcom \flinkp{a}{f}{b} = L$, $a \neq \pso$,
$a \neq \sso$, $b \neq \pso$ and $b \neq \sso$, \\
then $\getfield{t}{s}{f} \in \nosherr(L)$;
\end{itemize}
and the set $\secsherr(L)$ is inductively defined as follows:
\begin{itemize}
\item
if $L \dlcom \slinkp{s}{\sso} = L$, \\
then
$\setspot{t}{s},\, \equaltst{s}{t},\, \equaltst{t}{s},\, \undeftst{s},\,
 \addfield{s}{f},\, \rmvfield{s}{f},\, \hasfield{s}{f},\,
 \setfield{s}{f}{t},\, \setfield{t}{f}{s},\,
 \\ \phantom{\mbox{then }} \getfield{t}{s}{f} \in \secsherr(L)$;
\item
if $L \dlcom \slinkp{s}{a} \dlcom \flinkp{a}{f}{\sso} = L$, \\
then $\getfield{t}{s}{f} \in \secsherr(L)$.
\end{itemize}

The set $\shok$ contains all pairs $\tup{p,L} \in \TThr \x \DL$ such
that, if the first action that is performed by $p$ is an action of the
form $\dld.m$, where $m$ is a basic action of \DLD\ for changing the
content of a spot or field, the criterion for shedding of that spot or
field is met.
The general idea underlying the definition of $\shok$ given below is
that the criterion for shedding can be checked by mimicking shedding.
In checking, all possibilities must be considered:
\begin{itemize}
\item
if an action of the form $f.m$ with $f \neq \dld$ is encountered, then
two possibilities arise: (i)~the reply is $\True$ and (ii)~the reply is
$\False$;
\item
if an action of the form $\dld.m$ with $m$ a basic action of \DLD\ of
the form $\getatobj{s}$, $\setspot{s}{t}$, $\setfield{s}{f}{t}$ or
$\getfield{s}{t}{f}$ is encountered, then two possibilities arise:
(i)~the spot or field eligible for shedding is not shed and (ii)~the
spot or field eligible for shedding is shed.
\end{itemize}
In general, this means that many paths must be followed.
For regular threads, the number of paths to be followed will remain
finite and eventually either termination, deadlock, a mimicked primary
shedding error, a mimicked secondary shedding error or a cycle without
mimicked shedding errors will be encountered along each of the paths to
be followed.
The criterion for shedding is met if along each of the paths to be
followed it is not a mimicked primary shedding error that is encountered
first.
For non-regular threads, it is undecidable whether the criterion for
shedding is met.

The set $\shok$ is defined by $\shok = \shok'(1,\emptyset)$, where the
sets $\shok'(i,C) \subseteq \TThr \x \DLm$ for $i \in \set{0,1,2}$ and
$C \subseteq \TThr \x \DLm$ are defined by simultaneous induction as
follows:
\begin{itemize}
\item
$\tup{\Stop,L}, \tup{\DeadEnd,L} \in \shok'(i,C)$;
\item
if $f \neq \dld$,
\\ \phantom{if}
$\tup{p,L} \in
 \shok'(0,C \union \set{\tup{\pcc{p}{f.m}{q},L}})$,
\\ \phantom{if}
$\tup{p,L} \in
 \shok'(2,C \union \set{\tup{\pcc{p}{f.m}{q},L}})$,
\\ \phantom{if}
$\tup{q,L} \in
 \shok'(0,C \union \set{\tup{\pcc{p}{f.m}{q},L}})$,
\\ \phantom{if}
$\tup{q,L} \in
 \shok'(2,C \union \set{\tup{\pcc{p}{f.m}{q},L}})$,
\\
then $\tup{\pcc{p}{f.m}{q},L} \in \shok'(i,C)$;
\item
if $m \in \nosherr(L)$,
\\ \phantom{if}
$\use{(\pcc{p}{\dld.\mshv(i,m)}{q})}{\dld}{\DLDSm(L)} =
 \Tau \bapf (\use{r}{\dld}{\DLDSm(L')})$,
\\ \phantom{if}
$\tup{r,L'} \in
 \shok'(0,C \union \set{\tup{\pcc{p}{\dld.m}{q},L}})$,
\\ \phantom{if}
$\tup{r,L'} \in
 \shok'(2,C \union \set{\tup{\pcc{p}{\dld.m}{q},L}})$,
\\
then $\tup{\pcc{p}{\dld.m}{q},L} \in \shok'(i,C)$;
\item
if $m \in \secsherr(L)$,
\\
then $\tup{\pcc{p}{\dld.m}{q},L} \in \shok'(i,C)$;
\item
if $\rec{X}{E} \in \TThr$, $X = t_X \,\in\, E$,
$\tup{\rec{t_X}{E},L} \in \shok'(i,C)$,
\\
then $\tup{\rec{X}{E},L} \in \shok'(i,C)$;
\item
if $\tup{p,L} \in C$, then $\tup{p,L} \in \shok'(i,C)$.
\end{itemize}

In $\shok'(i,C)$, $i$ corresponds to the way in which a basic action of
\DLD\ for changing the content of a spot or field is dealt with in
checking:
\begin{itemize}
\item
without mimicking of its shedding if $i = 0$;
\item
with mimicking of its shedding by means of $\pso$ if $i = 1$;
\item
with mimicking of its shedding by means of $\sso$ if $i = 2$.
\end{itemize}
The members of $C$ correspond to the combinations of thread and state
encountered before in checking.
If such a combination is encountered again, this indicates a cycle
without shedding errors because a path is not followed further after
termination, deadlock, a mimicked shedding error or a cycle without
mimicked shedding errors has been encountered.

By the occurrence of the equation
$\use{(\pcc{p}{\dld.\mshv(i,m)}{q})}{\dld}{\DLDSm(L)} =
 \Tau \bapf (\use{r}{\dld}{\DLDSm(L')})$
in the third rule of the inductive definition of the sets $\shok'(i,C)$,
a service that is engaged in checking whether a pair
$\tup{p,L} \in \TThr \x \DL$ belongs to $\shok$ is close to reflecting
on itself.

Now, we are ready to define, for $L \in \DL \union \set{\undef}$, a data
linkage dynamic service $\DLDSs(L)$ that supports shedding.

Let $L \in \DL \union \set{\undef}$.
Then the \emph{shedding supporting data linkage dynamics service} with
initial state $L$, written $\DLDSs(L)$, is the service
$\tup{\DL \union \set{\undef},\effs,\ylds,L}$, where the functions
$\effs$ and $\ylds$ are the effect and yield functions satisfying the
equations in Table~\ref{eqns-DLDSs}.%
\begin{table}[!t]
\caption{Definition of effect and yield functions for \DLD\ with
  shedding}
\label{eqns-DLDSs}
\begin{eqntbl}
\begin{seqncol}
\effs(m,L,\pcc{p}{f.m}{q}) = \effop{\shv(m)}(L)
 & \mif m \in A_\DLD \land \tup{\pcc{p}{f.m}{q},L} \in \shok
\\
\effs(m,L,\pcc{p}{f.m}{q}) = \effop{m}(L)
 & \mif m \in A_\DLD \land \tup{\pcc{p}{f.m}{q},L} \not\in \shok
\\
\effs(m,L,\pcc{p}{f.m}{q}) = \undef
 & \mif m \not\in A_\DLD
\eqnsep
\ylds(m,L,\pcc{p}{f.m}{q}) = \yldop{\shv(m)}(L)
 & \mif m \in A_\DLD \land \tup{\pcc{p}{f.m}{q},L} \in \shok
\\
\ylds(m,L,\pcc{p}{f.m}{q}) = \yldop{m}(L)
 & \mif m \in A_\DLD \land \tup{\pcc{p}{f.m}{q},L} \not\in \shok
\\
\ylds(m,L,\pcc{p}{f.m}{q}) = \Blocked
 & \mif m \not\in A_\DLD
\eqnsep
\ylds(m,L,p) = \Blocked \limpl \effs(m,L,p) = \undef
\end{seqncol}
\end{eqntbl}
\end{table}

\section{Examples}
\label{sect-examples}

In this section, we give two examples that illustrate how the definition
of shedding supporting data linkage dynamics services can be used to
determine whether in a fixed case a spot or field, whose contents should
be changed, is shed.
The first example concerns a case where a spot is shed and the second
example concerns a case where a spot is not shed.

\begin{example}
Let
\begin{ldispl}
\begin{aeqns}
p   & = &
\dld(\getatobj{s}) \bapf
(\pcc{\Stop}{\dld(\getatobj{t})}{\DeadEnd})\;, \\
p'  & = &
\pcc{\Stop}{\dld(\getatobj{t})}{\DeadEnd}\;, \\
p'' & = &
\pcc{\Stop}{\dld(\setspot{t}{\sso})}{\DeadEnd}\;.
\end{aeqns}
\end{ldispl}%
Thread $p'$ is a residual thread of $p$ and thread $p''$ is $p'$ with
$\getatobj{t}$ replaced by $\setspot{t}{\sso}$ to mimic shedding.
Assume that the cardinality of $\AtObj$ is $1$, and let $a$ be the
unique atomic object such that $\AtObj = \set{a}$.
Then in $\use{p}{\dld}{\DLDSs(\emptydl)}$, spot $s$ is shed on
performing $\getatobj{s}$.
This is straightforwardly shown using the definition of shedding
supporting data linkage dynamics services.
It follows immediately from the definition of $\nosherr$ that:
\begin{ldispl}
(\getatobj{s}) \in \nosherr(\emptydl)\;, \\
(\getatobj{t}) \in \nosherr(\slink{s}{\pso})\;,
\end{ldispl}%
and it follows easily from the axioms for the use operators and the
definition of mimicking of shedding supporting data linkage dynamics
services that:
\begin{ldispl}
\use{p}{\dld}{\DLDSs(\emptydl)} =
\Tau \bapf (\use{p'}{\dld}{\DLDSs(\slink{s}{\pso})})\;, \\
\use{p'}{\dld}{\DLDSs(\slink{s}{\pso})} =
\Tau \bapf
(\use{\Stop}{\dld}{\DLDSs(\slinkp{s}{\pso} \dlcom \slinkp{t}{a})})\;, \\
\use{p''}{\dld}{\DLDSs(\slink{s}{\pso})} =
\Tau \bapf
(\use{\Stop}{\dld}{\DLDSs(\slinkp{s}{\pso} \dlcom \slinkp{t}{\sso})})\;.
\end{ldispl}%
Hence, by the definitions of $\shok'$ and $\shok$:
\begin{ldispl}
\tup{p',\slink{s}{\pso}} \in \shok'(0,\set{\tup{p,\emptydl}})\;, \\
\tup{p',\slink{s}{\pso}} \in \shok'(2,\set{\tup{p,\emptydl}})\;, \\
\tup{p,\emptydl} \in \shok'(1,\emptyset)\;, \\
\tup{p,\emptydl} \in \shok\;.
\end{ldispl}%
From this it follows by the definition of $\effs$ that
$\effs(\getatobj{s},\emptydl,p) = \effop{\sh(\getatobj{s})}(\emptydl)$.

On account of shedding, we have that
\begin{ldispl}
\use{p}{\dld}{\DLDSs(\emptydl)} = \Tau \bapf \Tau \bapf \Stop\;,
\end{ldispl}%
whereas
\begin{ldispl}
\use{p}{\dld}{\DLDS(\emptydl)} = \Tau \bapf \Tau \bapf \DeadEnd\;.
\end{ldispl}%
The point is that a positive reply is produced on performing
$\getatobj{t}$ only if spot $s$ is shed on performing $\getatobj{s}$.
\end{example}

\begin{example}
Let
\begin{ldispl}
\begin{aeqns}
p   & = &
\dld(\getatobj{s}) \bapf
(\pcc{(\dld(\setspot{u}{s}) \bapf \Stop)}
     {\dld(\getatobj{t})}{\Stop})\;, \\
p'  & = &
\pcc{(\dld(\setspot{u}{s}) \bapf \Stop)}
    {\dld(\getatobj{t})}{\Stop}\;, \\
p'' & = & \dld(\setspot{u}{s}) \bapf \Stop\;.
\end{aeqns}
\end{ldispl}%
Thread $p$ is the same thread as the one discussed in
Section~\ref{sect-shedding} and threads $p'$ and $p''$ are residual
threads of $p$.
Assume again that the cardinality of $\AtObj$ is $1$, and let $a$ be the
unique atomic object such that $\AtObj = \set{a}$.
Then in $\use{p}{\dld}{\DLDSs(\emptydl)}$, spot $s$ is not shed on
performing $\getatobj{s}$.
This is easily shown using the definition of shedding supporting data
linkage dynamics services.
It follows immediately from the definition of $\nosherr$, the definition
of $\secsherr$, and basic set theory that:
\begin{ldispl}
(\setspot{u}{s}) \not\in
 \nosherr(\slinkp{s}{\pso} \dlcom \slinkp{t}{a})\;, \\
(\setspot{u}{s}) \not\in
 \secsherr(\slinkp{s}{\pso} \dlcom \slinkp{t}{a})\;, \\
\tup{p'',\slinkp{s}{\pso} \dlcom \slinkp{t}{a}} \not\in
 \set{\tup{p,\emptydl},\tup{p',\slink{s}{\pso}}}\;,
\end{ldispl}%
and it follows easily from the axioms for the use operators and the
definition of mimicking of shedding supporting data linkage dynamics
services that:
\begin{ldispl}
\use{p}{\dld}{\DLDSs(\emptydl)} =
\Tau \bapf (\use{p'}{\dld}{\DLDSs(\slink{s}{\pso})})\;, \\
\use{p'}{\dld}{\DLDSs(\slink{s}{\pso})} =
\Tau \bapf
(\use{\Stop}{\dld}{\DLDSs(\slinkp{s}{\pso} \dlcom \slinkp{t}{a})})\;.
\end{ldispl}%
Hence, by the definitions of $\shok'$ and $\shok$:
\begin{ldispl}
\tup{p'',\slinkp{s}{\pso} \dlcom \slinkp{t}{a}} \not\in
\shok'(0,\set{\tup{p,\emptydl},\tup{p',\slink{s}{\pso}}})\;, \\
\tup{p',\slink{s}{\pso}} \not\in \shok'(0,\set{\tup{p,\emptydl}})\;, \\
\tup{p,\emptydl} \not\in \shok'(1,\emptyset)\;, \\
\tup{p,\emptydl} \not\in \shok\;.
\end{ldispl}%
From this it follows by the definition of $\effs$ that
$\effs(\getatobj{s},\emptydl,p) = \effop{\getatobj{s}}(\emptydl)$.
\end{example}

\section{Conclusions}
\label{sect-concl}

We have introduced shedding in the setting of data linkage dynamics and
have adapted the data linkage dynamics services described
in~\cite{BM08d} so that they support shedding.
The adaptation shows that the shedding feature is rather non-obvious.
In particular, it is striking how much the matter is complicated by
taking into consideration the semantic effects of the fact that the
number of data objects that can exist at the same time is always
bounded.

We consider the work presented in this paper a semantic validation of
shedding.
It is an entirely different question whether a real implementation of
shedding is of any use in practice.
We have not answered this question.
Empirical studies, see e.g.~\cite{HDH02a}, indicate that in general a
large part of the data objects that are reachable at a program point are
actually not used beyond that point.
However, the static approximation of shedding proposed in~\cite{KSK07a}
might be more useful in practice.

In the definition of shedding supporting data linkage dynamics services,
belonging to $\shok$ corresponds to meeting the criterion for shedding.
The set $\shok$ is defined using the idea of mimicked shedding.
As a result, the description of the criterion for shedding looks to be
rather concrete.
It is an open question whether a more abstract description of the
criterion for shedding can be given.
If so, our concrete description should be correct with respect to that
abstract description.

In the case of shedding, the use of forecasting turns out to be
semantically feasible.
No restrictions are needed to preclude forecasting from introducing
something paradoxical.
This is certainly not always the case.
For example, in the case of the halting problem, the use of forecasting
is not semantically feasible, see e.g.~\cite{BP04a}, and in the case of
security hazard risk assessment, the use of forecasting requires certain
restrictions, see e.g.~\cite{BBP07a}.

\bibliographystyle{spmpsci}
\bibliography{TA}

\end{document}